\documentclass[a4paper,UKenglish,cleveref, autoref, thm-restate]{lipics-v2021}

\usepackage{microtype} 
\usepackage{amsmath}
\usepackage{graphicx}
\usepackage[nocompress]{cite}
\usepackage{algorithm}
\usepackage[noend]{algorithmic}

\title{Computational Geometry with Probabilistically Noisy Primitive Operations} 

\author{David Eppstein}{University of California, Irvine \and\url{https://www.ics.uci.edu/~eppstein/}}{eppstein@uci.edu}{}{Research supported in part by NSF grant CCF-2212129.}

\author{Michael T. Goodrich}{University of California, Irvine \and\url{https://www.ics.uci.edu/~goodrich/}}{goodrich@uci.edu}{https://orcid.org/0000-0002-8943-191X}{Research supported in part by NSF grant CCF-2212129.}

\author{Vinesh Sridhar}{University of California, Irvine}{vineshs1@uci.edu}{https://orcid.org/0009-0009-3549-9589}{}

\authorrunning{D. Eppstein, M. T. Goodrich, and V. Sridhar}

\Copyright{David Eppstein, Michael T. Goodrich, and Vinesh Sridhar}

\ccsdesc[500]{Theory of computation~Computational geometry}
\ccsdesc[500]{Theory of computation~Random walks and Markov chains}

\keywords{Computational geometry, noisy comparisons, random walks} 

\category{} 

\relatedversiondetails[linktext={}]{Full Version}{https://arxiv.org/abs/2501.07707} 

\nolinenumbers 

\EventEditors{Pat Morin and Eunjin Oh}
\EventNoEds{2}
\EventLongTitle{Symposium on Algorithms and Data Structures (WADS 2025)}
\EventShortTitle{WADS 2025}
\EventAcronym{WADS}
\EventYear{2025}
\EventDate{August, 2025}
\EventLocation{Toronto, Ontario}
\EventLogo{}
\SeriesVolume{1}
\ArticleNo{1}

\begin{document}

\maketitle

\begin{abstract}
Much prior work has been done on designing computational geometry
algorithms that handle input degeneracies, data imprecision, and
arithmetic round-off errors. We take a new approach, inspired by
the \emph{noisy sorting} literature, and study computational geometry
algorithms subject to noisy Boolean primitive operations in which,
e.g., the comparison ``is point $q$ above line $\ell$?'' returns
the wrong answer with some fixed probability. We propose a novel
technique called \emph{path-guided pushdown random walks} that
generalizes the results of noisy sorting. We apply this technique
to solve point-location, plane-sweep, convex hulls in 2D and 3D,
and Delaunay triangulations for noisy primitives in optimal time
with high probability.
\end{abstract}

\section{Introduction}
In 1961, R{\'e}nyi~\cite{renyi} introduced a binary search problem
where comparisons between two values return the wrong answer independently with
probability $p<1/2$; see also, e.g., Pelc~\cite{PELC1989185,PELC200271}.
Subsequently, in 1994, 
Feige, Raghavan, Peleg, and Upfal~\cite{feige1994computing} showed
how to sort $n$ items in $O(n\log n)$ time, with high probability,\footnote{
  In this paper, we take ``with high probability'' (w.h.p.) to mean that the failure probability
  is at most $1/n^c$, for some constant $c\ge 1$.}
in the same noisy comparison model.
Building on this classic work, 
in this paper we study the design of efficient computational geometry algorithms
using Boolean geometric primitives, such as orientation queries or sidedness tests, that randomly return
the wrong answer (independently of previous queries and answers) with probability at most $p<1/2$ (where $p$ is known) and otherwise
return the correct answer.

The only prior work we are aware of
on computational geometry algorithms
tolerant to such non-persistent errors is work by 
Groz, Mallmann-Trenn, Mathieu, and Verdugo~\cite{benoit} on
computing a $d$-dimensional skyline of size $k$ and probability
$1-\delta$ in $O(nd\log(dk/\delta))$ time.
We stress that, as in our work, this noise model is different than 
the considerable prior work on geometric algorithms 
that are tolerant to uncertainty,
imprecision, or degeneracy in their inputs, some of which we review below.
Indeed, the motivation for our study does not come from issues that arise,
for example, from round-off errors, geometric measurement errors, or degenerate geometric configurations.
Instead, our motivation comes from potential applications involving quantum computing, where
a quantum computer is used to answer primitive queries, which return an incorrect
answer with a fixed probability at most $p<1/2$; 
see, e.g.,~\cite{allcock2023quantumtimecomplexitydivide,iwama,math11224707}. 
We see our contribution as a complement to prior work
that applies quantum algorithms in a black-box
fashion to quickly solve geometric
problems, e.g., see~\cite{ambainis2020quantum}.
Our technique gives greater flexibility to algorithms 
that rely on noisy primitives or subprocessors and may be useful in 
further developing this field.
A second motivation for this noise model comes from communication complexity, in which hash 
functions can be used to reduce communication 
costs in distributed algorithms at the 
expense of introducing error. 
Viola models this
behavior with noisy primitives to efficiently construct higher-level,
fault-tolerant algorithms~\cite{Vio-Comb-15}.

A simple observation, made 
for sorting by 
Feige, Raghavan, Peleg, and Upfal~\cite{feige1994computing},
is that we can use any polynomial-time 
algorithm based on correct
primitives by repeating each primitive operation 
$O(\log n)$ times and taking the majority answer as the result.  
This guarantees correctness w.h.p., 
but increases the running time by a $\log n$ factor. 
In this paper, we design
computational geometry algorithms with noisy primitive operations
that are correct w.h.p. without incurring this overhead. 
In \cref{sec:lower-bounds},
we show that the logarithmic overhead is unavoidable for certain problems, 
including closest pairs and detecting collinearities.

\subsection{Related Work}

There is considerable prior work on sorting and searching
with noisy comparison errors.
For example, 
Feige, Raghavan, Peleg, and Upfal~\cite{feige1994computing} 
show that one can sort
in $O(n\log n)$ time w.h.p.~with
probabilistically noisy comparisons.
Dereniowski, Lukasiewicz, and Uznanski~\cite{dereniowski2023} study
noisy binary searching, deriving time bounds for 
constant factors involved.
A similar study has also been done by
Wang, Ghaddar, Zhu, and Wang~\cite{wang24}. 
Klein, Penninger, Sohler, and Woodruff solve linear programming in 2D under a different model of primitive errors~\cite{tolerant}, in which errors persist regardless of whether primitives are recomputed.

Other than the work by Groz {\it et al.}~\cite{benoit} mentioned above,
we are not aware of prior work on computational
geometry algorithms with random, independent noisy primitives. Nevertheless, considerable prior work has designed algorithms
that can deal with input degeneracies, data imprecision, and 
arithmetic round-off errors.
For example,
several researchers have studied general methods for dealing with 
degeneracies in inputs to geometric algorithms, 
e.g., see~\cite{yap2006symbolic,edelsbrunner1990simulation}.
Researchers have designed
algorithms for geometric objects with imprecise positions,
e.g., see~\cite{loffler2009data,loffler2008delaunay}.
In addition, significant prior work has dealt with
arithmetic round-off errors and/or performing geometric primitive operations
using exact arithmetic,
e.g., see~\cite{fortune1996static,mehlhorn2001exact}.
While these prior works have made important contributions to 
algorithm implementation in 
computational geometry, they
are orthogonal
to the probabilistic noise model we consider in this paper.

Emamjomeh-Zadeh, Kempe, and Singhal~\cite{emamjomeh2016deterministic}
and Dereniowski, Tiegel, Uzna{\'n}ski, 
and Wolleb-Graf~\cite{dereniowski2018framework} explore a generalization of noisy binary search to graphs,
where one vertex in an undirected, positively weighted graph is a target. 
Their algorithm identifies the target by adaptively querying vertices. A query to a node $v$ either determines 
that $v$ is the target or produces an edge out of $v$ 
that lies on a shortest path from $v$ to the target.
As in our model, the response to each query is wrong independently
with probability $p<1/2$.
This problem is different than the graph search we study
in this paper, however, 
which is better suited to applications in computational geometry.
For example, in computational geometry applications,
there is typically a search path, $P$, that needs to be traversed to a 
target vertex, but the search path $P$ need not be a shortest path.
Furthermore, in such applications,
if one queries using a node, $v$, 
that is not on $P$, it may not even be possible 
to identify a node adjacent to $v$ that is closer to the target vertex.

Viola~\cite{Vio-Comb-15} uses a technique similar to ours to handle errors in a communication protocol. In one problem studied by Viola, two participants with $n$-bit values seek to determine which of their two values is largest. This can be done by a noisy binary search for the highest-order differing bit position. Each search step performs a noisy equality test on two prefixes of the inputs, by exchanging single-bit hash values. The result is an $O(\log n)$ bound on the randomized communication complexity of the problem. Viola uses similar protocols for other problems including testing whether the sum of participant values is above some threshold. The noisy binary search protocol used by Viola directs the participants down a decision tree, with an efficient method to test whether the protocol has navigated down the wrong path in order to backtrack. One can think of our main technical lemma as a generalization of this work to apply to any DAG. 

\subsection{Our Results}
This work centers around our novel technique, 
\emph{path-guided pushdown random walks}, described in \cref{sec:random}. 
It extends
noisy binary search in two ways: it can handle
searches where the decision structure of comparisons is in general
a DAG, not a binary tree, and it also correctly returns a non-answer
in the case that the query value is not found. These two traits
allow us to implement various geometric algorithms in the noisy
comparison setting.

However, to apply path-guided pushdown random walks, one must design an oracle that, given a sequence of comparisons, can determine if one of them is incorrect in constant time. Because different geometric algorithms use different data structures and have different underlying geometry, we must develop a unique oracle for each one. The remainder of the paper describes  noise-tolerant implementations with optimal running times for plane-sweep algorithms, point location,
convex hulls in 2D and 3D, and Delaunay triangulations. We also present a dynamic 2D hull construction that runs in $O(\log^2 n)$ time w.h.p. per operation. See \cref{tab:summary} for a summary of our results.

Our algorithms to construct the trapezoidal decomposition, Delaunay triangulation, and 3D convex hull are adaptions of classic randomized incremental constructions. 
A recent paper by Gudmundsson and Seybold~\cite{gudmundsson2022tail} published at SODA 2022 claimed that such constructions produce search structures of depth $O(\log n)$ and size $O(n)$ w.h.p. in $n$, implying that the algorithms run in $O(n\log n)$ time w.h.p.
Unfortunately, the authors have contacted us privately to say that their analysis of the size of such structures contains a bug. 
In this paper, we continue to use the standard expected time analysis. 

\begin{table*}[t]
  \centering
  \begin{tabular}{|c|c|c|}
    \hline
    Algorithm & Runtime & Section \\
    \hline
    Trapezoidal Map & $\Theta(n\log n)^{\star}$ & \cref{sec:TrapMap} \\
    Trapezoidal Map with $k$ Crossings & $O((n+k)\log n)$ & \cref{sec:TrapMapCrossings}\\
    2D Closest Points & $\Theta(n\log n)^\dagger$ & \cref{sec:closest}\\
    2D Convex Hull & $\Theta(n\log n)^\dagger$ & \cref{sec:2DHull}\\
    Dynamic 2D Convex Hull & $O(\log^2 n)$ per update & \cref{sec:Dyn2D}\\
    3D Convex Hull & $\Theta(n\log n)^{\star\dagger}$ & \cref{sec:3D-hull}\\
    Delaunay Triangulation & $\Theta(n\log n)^{\star}$ & \cref{sec:DT}\\
    \hline
  \end{tabular}
  \caption{Our main results in the noisy setting. All algorithms succeed w.h.p. in $n$. Runtimes marked with a $^\star$ are in expectation. Runtimes marked with a $^\dagger$ are optimal despite having faster algorithms in the non-noisy setting.}
  \label{tab:summary}
\end{table*}

\section{Preliminaries}
\paragraph*{Noisy Boolean Geometric Primitive Operations}
\label{sec:prelim-noise}

Geometric algorithms typically rely on one or more Boolean 
geometric primitive operations that are assumed to be computable in $O(1)$ time. For example, in a Delaunay triangulation algorithm, this may be determining if some point $p$ is located in some triangle $\Delta$; in a 2D convex hull algorithm, this may be an orientation test; etc. 
Here we assume that a geometric algorithm
relies on primitive operations that each outputs a 
Boolean value and has a fixed probability $p < 1/2$ of outputting the wrong answer. 
As in earlier work for the sorting problem~\cite{feige1994computing}, 
we assume non-persistent errors, in which each primitive test can be
viewed as an independent weighted coin flip. 

In each section below, we specify the Boolean geometric primitive operation(s)
relevant to the algorithm in consideration. We note here
that, while determining whether two objects $a$ and $b$ have equal
value may be a noisy operation, determining whether two pointers
both point to the same object $a$ is not a noisy operation.
We also note that manipulating and comparing non-geometric data,
such as pointers or metadata of nodes in a tree (e.g., for rotations), 
are not noisy operations.
This is true even if the tree was constructed
using noisy comparisons.

\paragraph*{The Trivial Repetition Strategy}
\label{sec:repetition}
As mentioned above,
Feige, Raghavan, Peleg, and Upfal~\cite{feige1994computing} observed
in the context of the noisy sorting problem
that by simply repeating a primitive operation $O(\log n)$ times
and choosing the decision returned a majority of the time, one can
select the correct answer w.h.p. The constant in this logarithmic
bound can be adjusted as necessary to make a polynomial number of
correct decisions, w.h.p., as part of any larger algorithm.  Indeed,
this naive method immediately implies $O(n\log^2 n)$ algorithms
for a majority of the geometric constructions we discuss below. The goal of
our paper is to improve this to an optimal running time using the
novel technique described in \cref{sec:random}.

\paragraph*{General Position Assumptions}
For the sake of simplicity of expression,
we make standard general position assumptions throughout this paper: no two segment endpoints in a trapezoidal decomposition and no two events in a plane sweep algorithm have the same $x$-coordinate, 
no three points lie in a line and no four points lie on a plane for 2D and 3D convex hulls respectively, and no four points lie in a circle for Delaunay triangulations. Applying perturbation methods~\cite{mehlhorn2006reliable} or implementing special cases in each algorithm would allow the relaxation of these assumptions.

\section{Path-Guided Pushdown Random Walks}
\label{sec:random}
\begin{figure}[htb!]
    \centering
    \includegraphics[width=0.4\linewidth]{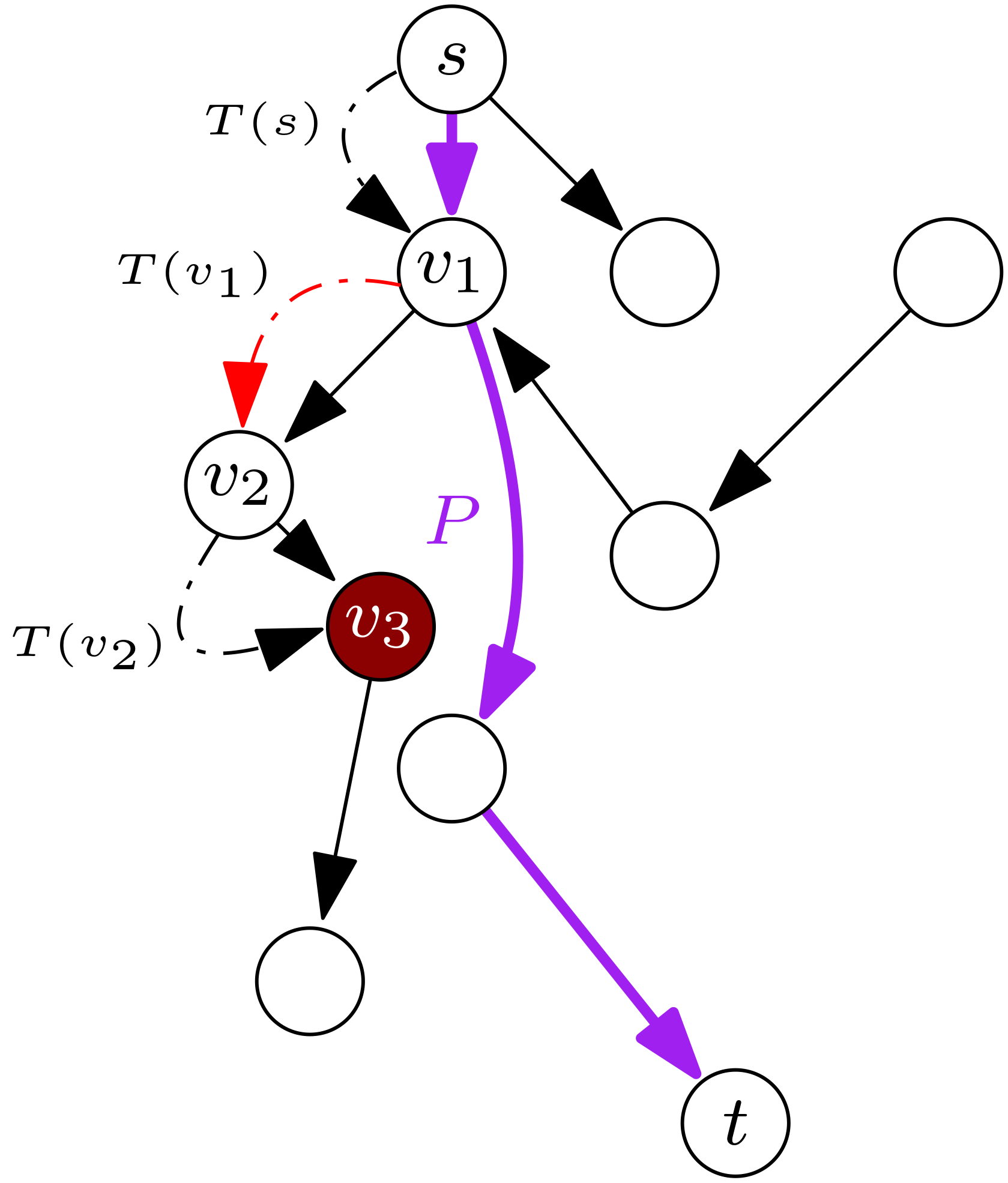}
    \caption{Here we show a sample execution of 
    path-guided pushdown random walks in
    some DAG $G$. 
    The transition oracle
    acts \emph{arbitrarily} when it lies, so we may
    end up with the sequence shown here, a correct
    move prior to our current node but
    an incorrect move to $v_2$. 
    To apply path-guided pushdown random walks, we must be able to
    determine whether we are on $P$ in $O(1)$ comparisons, no matter where
    we are in $G$.}
    \label{fig:random-walk}
\end{figure}

In this section, we provide an analysis tool that we use repeatedly in
this paper and which may be of independent interest (e.g., to analyze
randomized routing protocols).
Specifically, we introduce \emph{path-guided pushdown random walks}, which 
are related to biased random walks on a graph 
(e.g., see~\cite{azar1996biased,fronczak2009biased}) and
generalize the noisy
binary search problem~\cite{feige1994computing,geissmann2019optimal}. See \cref{fig:random-walk} for a depiction of path-guided pushdown random walks.

A path-guided pushdown random walk is defined in terms of a directed
acyclic graph (DAG), $G$, 
that has a
starting vertex, $s$, a target vertex, $t$, and a path,
$P=(s,v_1,v_2,\ldots,t)$, from $s$ to $t$, in $G$
(we assume $s\not=t$).
We start our walk from the start vertex, $v\leftarrow s$, and we use a
stack, $S$, which initially contains only $s$, and a transition
oracle, $T(v)$, to determine our walk.
For each vertex, $v$, during our walk, we consult
the transition oracle, $T(v)$, which first tells 
us whether $v\in P$ and if so, then $T(v)$ tells us the next vertex in $P$ to
visit to make progress towards $t$. $T(v)$ can return $v$, which means
we should stay at $v$, e.g., if $v=t$.

Our model allows $T$ to ``lie.''
We assume a
fixed error probability,\footnote{The threshold of $1/15$
  simplifies our proof. 
  We can tolerate any higher error probability bounded below $1/2$,
  by repeating any query a constant number of times
  and taking the majority answer.}
$p_e<1/15$,
such that $T$ gives the correct answer with probability $1-p_e$, 
independently each time we query $T$.
With probability $p_e$, $T(v)$ can lie, i.e., $T(v)$
can indicate ``$v\in P$'' when $v\not\in P$,
$T(v)$ can indicate ``$v\not\in P$'' when $v\in P$, or $T(v)$ can return
a ``next'' vertex that is not an actual next 
vertex in~$P$ (including returning $v$ itself even though $v\not=t$).
Importantly, this next vertex must be an outgoing neighbor of $v$.
This allows us to maintain the 
invariant that $S$ holds an actual path in~$G$ (with repeated vertices).
Our traversal step, for current vertex $v$, is as follows:

\begin{itemize}
\item
If $T(v)$ indicates that $v\not\in P$ (and $v\not=s$), 
then we set $v\leftarrow S.{\sf pop}()$,
which may be $v$ again, for the next iteration.
This is a backtracking step.
\item
If $T(v)$ indicates that $v\in P$,
then let $w$ be the vertex indicated by $T(v)$ as next in $P$,
such that $v=w$ or $(v,w)$ is an edge in $G$.\footnote{
  If $w\not=v$ and $(v,w)$ is not an edge in $G$, 
  we immediately reject this call to $T$ and repeat the call to~$T$.}
\begin{itemize}
\item
If $v=w$, then we call $S.{\sf push}(v)$ and repeat the iteration with this
same $v$, since this is evidence we have reached the target, $t$.
\item
Else ($v\not=w$) if $v=S.{\sf top}()$, 
then we set $v\leftarrow S.{\sf pop}()$. That is, in this case, we don't
immediately transition to $w$, but we take this as evidence that we should
not stay at $v$, as we did in the previous iteration. This is another 
type of backtracking step.
\item
Otherwise, we call $S.{\sf push}(v)$ and set $v\leftarrow w$ for 
the next iteration.
\end{itemize}
\end{itemize}

We repeat this step until we are confident we can stop, which occurs when enough copies of the same vertex occur at the top of the stack.

\begin{theorem}
\label{thm:path}
Given an error tolerance, $\varepsilon=n^{-c}$ for $c>0$,
and
a DAG, $G$, with a path, $P$, from a vertex, $s$, to a distinct vertex, $t$,
the path-guided pushdown random walk in $G$ starting from $s$ will terminate
at $t$ with probability
at least $(1-\varepsilon)$ after $N=\Theta(|P|+\log (1/\varepsilon))$
steps, for a transition oracle, $T$, with error probability $p_e< 1/15$.
\end{theorem}

The proof of \cref{thm:path} can be found in \cref{sec:dag-lemma-pf}. In short, we show that each ``good'' action by $T$ can undo any ``bad'' action by $T$. By applying a Chernoff bound, we then show that, w.h.p., we reach $t$ after $N$ steps and terminate after another $O(\log (1/\varepsilon))$ steps. The requirement that $\varepsilon$ be polynomially small is used to ensure that premature termination is unlikely.  For the remainder of the paper, we assume that $\varepsilon = 1/n^c$ for some constant $c>0$. 
A majority of the algorithms below invoke path-guided pushdown random walks at most $O(n)$ times. 
Taking the union bound over all $O(n)$ invocations still shows that path-guided pushdown random walks 
fails with probability $O(1/n^{c-1})$. 
Thus, for $c \geq 2$, all invocations succeed w.h.p. 
This choice of $\varepsilon$ adds $O(\log n^c) = O(\log n)$ to the time for each walk. 
However, in our applications it can be shown that $|P| = O(\log n)$, 
so this does not change the asymptotic complexity of the operation.

\section{Noisy Randomized Incremental Construction for Trapezoidal Maps}\label{sec:TrapMap}

\begin{figure}[t]
\centering
\begin{subfigure}{.5\textwidth}
  \centering
  \includegraphics[width=.9\linewidth]{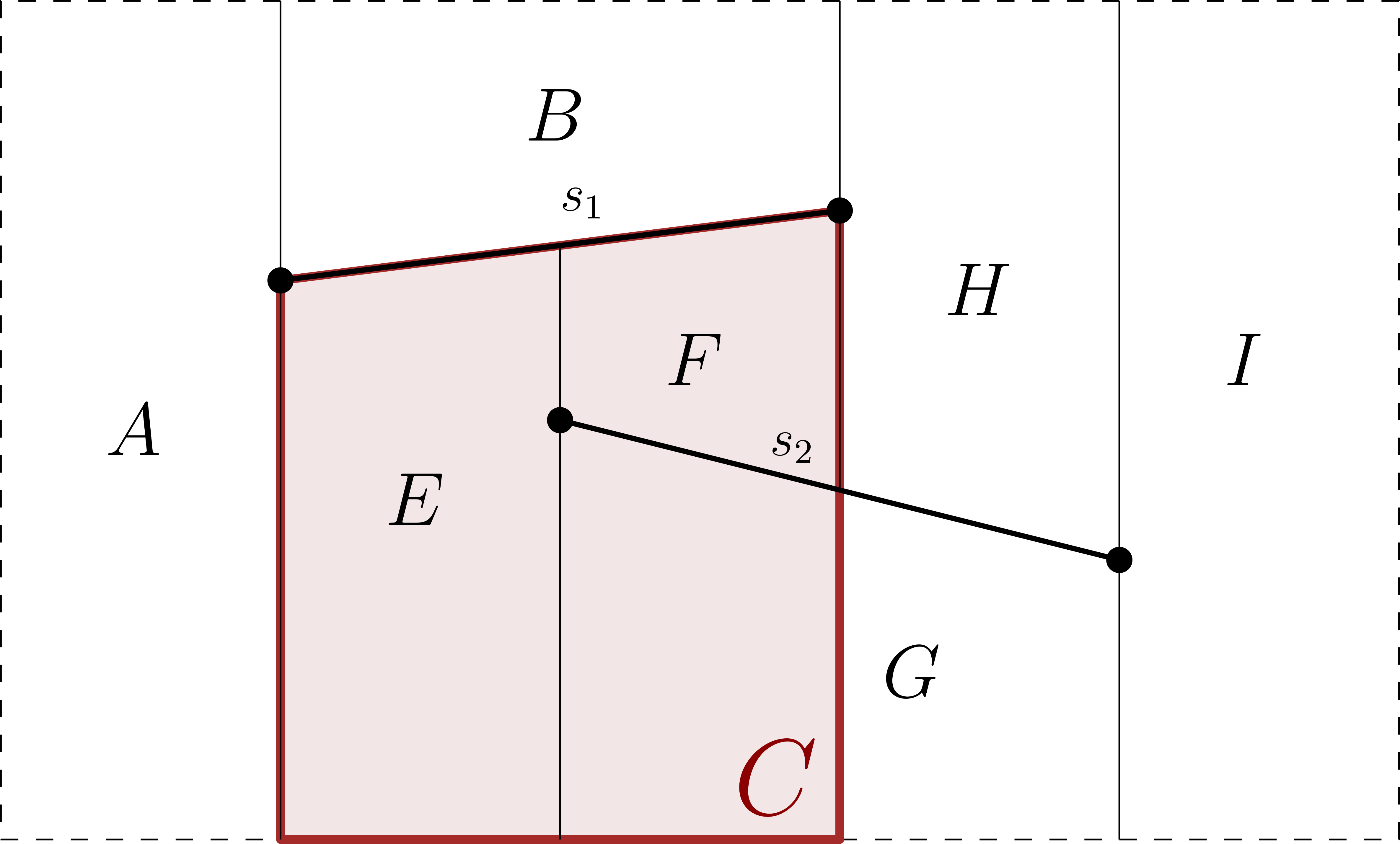}
\end{subfigure}%
\begin{subfigure}{.5\textwidth}
  \centering
  \includegraphics[width=.9\linewidth]{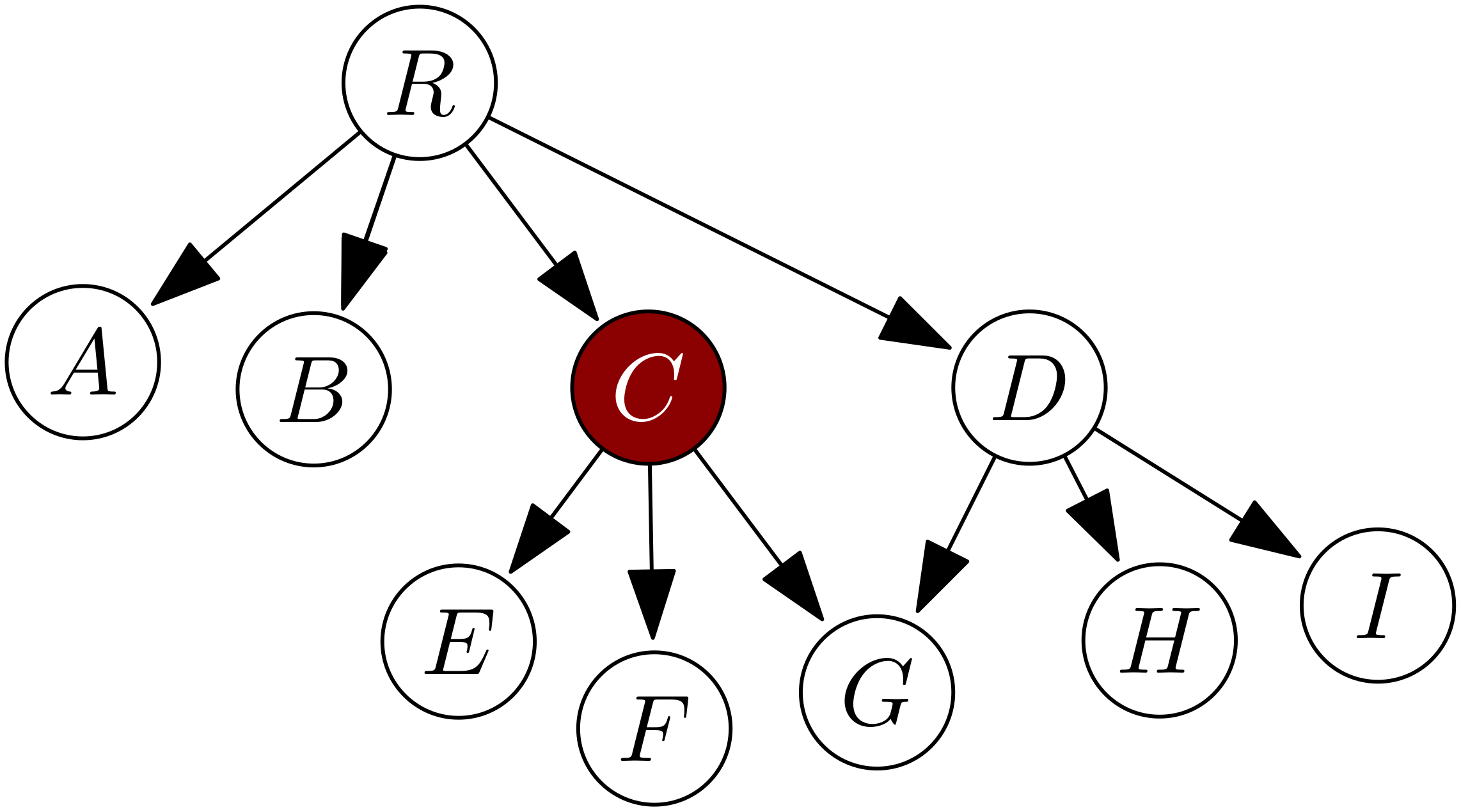}
\end{subfigure}
\caption{
  An instance of a path-guided pushdown random walk in 
  the trapezoidal map history DAG.
  $R$ represents the initial bounding box, the leaf nodes represent currently visible trapezoids.
  The remaining nodes represent destroyed trapezoids.
  We are at node $C$ in a path-guided pushdown random walk, 
  so we test if our query point
  lies in trapezoid $C$, the shaded region in the
  decomposition.
}
\label{fig:trap-map-oracle}
\end{figure}

In this section, we show how a history DAG in 
a randomized incremental construction (RIC) algorithm
can be used as the DAG of \cref{thm:path}.
%
Suppose we are given a set of $n$ non-crossing line segments in the plane
and wish to construct their trapezoidal decomposition, $\mathcal{D}$.
We outline below 
how the history
DAG for a (non-noisy) RIC 
algorithm for constructing $\mathcal{D}$
can be used as the DAG of \cref{thm:path} (see also \cref{sec:RIC-expl}).
We show that, even in the noisy setting, such a DAG can
perform point location in the
history of trapezoids created and/or destroyed 
within the algorithm to locate the endpoints of 
successively inserted line segments.
In particular,
we consider a history DAG 
where any trapezoid destroyed in a given iteration points to the new trapezoids that replace it.
This is in contrast to another variant of a history DAG that
represents the segments themselves as nodes in
the DAG~\cite{decomputational}, which seems less usable when
primitives are noisy.
To make our RIC algorithm noise-tolerant, we must solve two issues. 
The first is to navigate the history DAG in $O(\log n)$ time w.h.p.; 
the second is to walk along each segment to 
merge and destroy trapezoids when it is added.

To navigate down the history DAG, we apply a path-guided pushdown random walk. 
To do so we must test, with a constant number of operations, 
whether we are on the path that a non-noisy algorithm would 
search to find the query endpoint $q$ of a new segment, $s_i$. 
See \cref{fig:trap-map-oracle} for an example of this process. 
Each node of our history DAG represents a trapezoid 
(either destroyed if an internal node or current if a leaf node). 
Importantly, each of the at most four children 
of a node cover the parent's trapezoid, but do not overlap. 
Therefore, we can define a unique path, $P$ w.r.t 
query $q$, to be a sequence of trapezoids 
beginning at the root of the history DAG and ending at the leaf-node
trapezoid containing $q$. 
Each node on the path is the unique child of the 
previous node that contains $q$. 
Thus, all our transition oracle must do to determine if we are 
on the correct path is to check 
whether $q$ is contained by the trapezoid mapped to our current node $v$.
%
If this test succeeds, the oracle determines (rightly or wrongly) that $v$ is on a valid path, and it proceeds to compare~$q$ against the segment whose addition split trapezoid $v$ in order to return the next node of the walk. 
If one or more of these tests fails, the oracle says that $v$ is not on a valid path. 
Let $\ell_i$ be the sum of the lengths of the two unique paths corresponding to the endpoints of $s_i$. If we set our error probability for path-guided pushdown random walks to be w.h.p. in $n$, the point-location cost to insert the $i$th segment is $O(\ell_i + \log n)$. 

For the second issue, suppose there are $d_i$ trapezoids between the
left and right endpoint of the $i$th segment to be inserted,
and that we need to walk left-to-right in the current subdivision to find them.
To find the next
trapezoid in this walk, we simply test if the segment endpoint that
defines the right boundary of the current trapezoid lies above or below segment $s_i$, 
e.g., determining whether to choose its upper-right or lower-right
neighbor.  Combining this above-below test with the trivial
repetition strategy from \cref{sec:repetition}, we can compute the
correct sequence of trapezoids in $O(d_i\log n)$ time w.h.p. 

Via the standard backwards analysis, $\sum_{i=1}^n \ell_i = O(n\log n)$ and $\sum_{i=1}^n d_i = O(n)$ in expectation~\cite{decomputational}. 
For completeness, we defer such details to~\cref{sec:RIC-expl}. 
It has also been shown~\cite{decomputational,sen2019unified} that the search depth of the final history DAG is $O(\log n)$ w.h.p. Therefore, after constructing the decomposition, we can use path-guided pushdown random walks to answer planar point-location queries in $O(\log n)$ time w.h.p. We conclude the following.
%

\begin{theorem}
\label{thm:trap}
We can successfully compute a trapezoidal decomposition map of $n$ non-crossing line 
segments w.h.p. in expected $O(n\log n)$ time, even with noisy primitives,
and thereby construct a data structure (the history DAG) that answers
point location queries in $O(\log n)$ time w.h.p.
\end{theorem}

It is natural to hope that this can be extended to line segment arrangements with crossings, for which the best non-noisy time bounds are $O(n\log n+k)$, achieved with a similar randomized incremental approach. 
However, as we explain in \cref{sec:TrapMapCrossings} and \cref{sec:arrangement-lb}, an extra logarithmic factor may be necessary.

\section{Plane-Sweep Algorithms}
In this section, we show that many plane-sweep algorithms
can be adapted to our noisy-primitive model.

\subsection{Noisy Balanced Binary Search Trees}
\label{sec:oracle-BST}
We begin by noting that we can implement binary search trees
storing geometric objects to support searches and updates in 
$O(\log n)$ time w.h.p. by
adapting the noisy balanced binary search trees recently
developed for numbers in quantum applications by
Khadiev, Savelyev, Ziatdinov, and Melnikov~\cite{math11224707}
to our noisy-primitive framework for geometric objects.
Plane-sweep algorithms require us to store an ordered sequence 
of events that are visited by the sweep-line throughout 
the course of the algorithm while maintaining an active set of 
geometric objects for the sweep line in a balanced binary tree. 
Some plane-sweep algorithms, such as
the one of Lee and Preparata that divides a polygon into monotone
pieces~\cite{lee1976location}, have a static set of events that can
simply be maintained as a sorted array that is constructed
using an existing noisy sorting algorithm~\cite{feige1994computing}.
Others, however, add
or change events to the event queue as the algorithm progresses. 
In addition, the
algorithms must also maintain already-processed data in a way that
allows for new data swept over to be efficiently incorporated. The
two most common dynamic data structures in plane-sweep algorithms
are dynamic binary search trees and priority queues. 
Throughout this section, noise is associated with the comparison 
``$a \leq b$?'' where $a$ and $b$ are geometric objects.

Once a data structure is built, its underlying structure can be
manipulated without noise. For example, we need no knowledge of the
values held in a tree to recognize that it is imbalanced and to
perform a rotation, allowing implementation of self-balancing binary
search trees. 
For example,
Khadiev, Savelyev, Ziatdinov, and Melnikov~\cite{math11224707}
show how to implement
red-black trees~\cite{guibas1978dichromatic} in the noisy
comparison model for numbers. 
We observe here that the same method
can be used for ordered geometric objects compared with noisy
geometric primitives,
with a binary search tree serving as the search DAG 
and a root-to-leaf search path as the path.
See \cref{sec:BST-appendix} for an explanation for how to instantiate 
path-guided pushdown random walks on a binary search tree. 

\subsection{Noisy Trapezoidal Decomposition with Segment Crossings}\label{sec:TrapMapCrossings}
We can implement an event queue
as a balanced binary search tree with a pointer to the smallest element to perform priority queue operations in the noisy setting. 

\begin{theorem}
\label{thm:segments}
Given a set, $S$, of $n$ 
$x$-monotone
pseudo-segments\footnote{A set of $x$-monotone psuedo-segments
   is a set of $x$-monotone curve segments that do not self-intersect
   and such that any two of them intersect at most 
   once~\cite{chan2003levels,agarwal2005pseudo}.
   }
 in the plane, we can construct a trapezoidal
map of $S$ in $O((n+k)\log n)$ time w.h.p., where $k$ is the number
of pairs of crossing segments, even with noisy primitive operations.
\end{theorem}
\begin{proof}
Bentley and Ottmann~\cite{bentley1979algorithms} compute a trapezoidal map of 
line segments $S$ via a now well-known
plane-sweep algorithm, which translates directly to $x$-monotone 
pseudo-segments. 
Pseudo-segment endpoints and intersection points are kept
in an event queue ordered by $x$-coordinates and line segments intersecting
the sweep line are kept in a balanced binary search tree, both
of which can be implemented to support searches and updates in $O(\log n)$
time w.h.p. in the noisy model, as described above.
Given the $O((n+k)\log n)$-time performance of the Bentley-Ottmann
algorithm, this implies that we can construct the trapezoidal map
of $S$ in $O((n+k)\log n)$ time w.h.p.
\end{proof}

We note that this running time 
matches the construction time in \cref{thm:trap} for non-crossing
line segments, but it does
not give us a point-location data structure as in \cref{thm:trap}.
Further,
the upper bound of \cref{thm:segments} is optimal
with noisy primitive operations for $k=\Theta(n^2)$
via a reduction from computing line arrangements to computing trapezoidal decompositions.
We show that computing an arrangement of $n$ lines takes $\Omega(n^2\log n)$ time in  Appendix~\ref{sec:arrangement-lb}.




\subsection{Noisy Closest Pair}
\label{sec:closest}

Here, we present another plane-sweep application in the noisy setting.

\begin{theorem}
We can find a closest pair of points, from $n$ points in the plane, with noisy primitives, in time $O(n\log n)$ w.h.p.
\end{theorem}

\begin{proof}
Hinrichs, Nievergelt, and Schorn show how to find a pair of closest points 
in the plane in $O(n\log n)$ time~\cite{hinrichs1988plane} by
sorting the points by their $x$-coordinate and then plane-sweeping 
the points in that order. They maintain the minimum distance $\delta$ seen so far, and a \emph{$y$-table} of \emph{active points}. A point is active if has been processed and its $x$-coordinate is within $\delta$ of the current point; once this condition stops being true, it is removed from the $y$-table. The $y$-table may be implemented using a balanced binary tree. When a point is processed, the $y$-table is updated to remove points that have stopped being active, by checking the previously-active points sequentially according to the sorted order until finding a point that remains active. The new point is inserted into the $y$-table, 
and a bounded number of its nearby points in the table are selected. The distances between the new point and these selected points are compared to $\delta$, and $\delta$ is updated if a smaller distance is found.

In our noisy model, sorting the points takes $O(n\log n)$ time w.h.p.~\cite{feige1994computing}. Checking whether a point has stopped being active and is ready to be removed from the $y$-table may be done using the trivial repetition strategy of \cref{sec:repetition}; its removal is a non-noisy operation. Inserting each point into the $y$-table takes $O(\log n)$ time w.h.p. Selecting a fixed number of nearby neighbors is a non-noisy operation, and comparing their distances to $\delta$ may be done using the trivial repetition strategy with a noisy primitive that compares two distances determined by two pairs of points. In this way, we perform $O(\log n)$ work for each point when it is processed, and $O(\log n)$ work again later when it is removed from the $y$-table. Overall, the time is $O(n\log n)$ w.h.p.
\end{proof}

We give this result mostly as a demonstration for performing
plane sweep in the noisy model, since
a closest pair (or all nearest neighbors) may be found by constructing the Delaunay triangulation (\cref{sec:DT}) and then performing
min-finding operations on its edges.

An algorithm by Rabin \cite{rabin1976probabilistic} finds 2D closest points in $O(n)$ time, beating the $O(n\log n)$ plane sweep implementation that our solution is based on, in a model of computation allowing integer rounding of numerical values derived from input coordinates. However, computing the minimum of $n$ elements takes $O(n\log n)$ time for noisy comparison trees \cite{feige1994computing}, and Rabin's algorithm includes steps that find the minimum among $O(n)$ distances, so it is not faster than our algorithm in our noisy model. 
In \cref{sec:closest-lb}, we prove that
if data is accessed only through noisy primitives that combine information from $O(1)$ data points (allowing integer rounding), then finding the closest pair w.h.p. requires time $\Omega(n\log n)$. Thus, the plane-sweep closest-pair algorithm 
is optimal in this model.

\section{Convex Hulls}
In this section, we describe algorithms for constructing
convex hulls that can tolerate 
probabilistically noisy primitive operations.
Here, primitive operations are orientation tests and visibility 
tests for 2D and 3D convex hulls respectively. Sorting is also
used throughout, so comparing the $x$ or $y$ coordinates of
two points is also a noisy primitive.

\subsection{Static Convex Hulls in 2D}\label{sec:2DHull}
We begin by observing that it is easy to construct
two-dimensional convex hulls in $O(n\log n)$ time w.h.p.
Namely, we
simply sort the points using a noise-tolerant sorting algorithm~\cite{feige1994computing}
and then run Graham scan~\cite{decomputational}. 
The Graham scan phase of the algorithm
uses $O(n)$ calls to primitives, 
so we can simply use the trivial repetition strategy from \cref{sec:repetition}
to implement this phase in $O(n\log n)$ time w.h.p.
We note that this is the best possible, however,
since just computing the minimum or maximum of $n$ values
has an $\Omega(n\log n)$-time lower bound for a high-probability bound
in the noisy primitive model~\cite{feige1994computing}.

\subsection{Dynamic Convex Hulls in 2D}\label{sec:Dyn2D}

Overmars and van Leeuwen~\cite{overmars1981maintenance}
show how to maintain
a dynamic 2D
convex hull with $O(\log ^2 n)$ insertion and deletion times, 
$O(\log n)$ query time, and $O(n)$ space.
They maintain a binary search tree, $T^*$,
of points stored in the leaves ordered by $y$-coordinates 
and augment each internal node with
a representation of half the convex hull of the points in that
subtree. 
They define an $lc$-hull of~$P$ as the convex hull
of $P \cup \{(\infty, 0)\}$, which is the left half of the hull of
$P$ ($rc$-hull is defined symmetrically). It is easy to see that
performing the same query on both the $lc$- and $rc$-hull of $P$ 
allows us to determine if a query point is inside, outside,
or on the complete convex hull. 
We follow this same approach, showing
how their approach can be implemented in the noisy-primitive 
model using path-guided pushdown random walks.

Overmars and van~Leeuwen show that two $lc$-hulls $A$ and $C$
separated by a horizontal line can be merged in $O(\log n)$ time.
Rather than being a single binary search, however, their method
involves a joint binary search on the two trees
that represent $A$ and $C$ with the aim of finding a tangent line
that joins them. 
At each step, we follow one of ten cases to determine
how to navigate the trees of $A$ and/or $C$. 
Because each decision may only advance in one tree,
noisy binary search is not sufficient. 
Fortunately, as we show, the path-guided
pushdown random walk framework is powerful enough to solve this
problem in the noisy model.
The challenge, of course, is to define an appropriate transition oracle.
Say that the transition oracle has the ability to identify which
case corresponds to our current position in~$A$ and~$C$. This
decision determines which direction to navigate in~$A$ and/or~$C$ and
corresponds to choosing a child to navigate down an implicit
decision DAG wherein each node has ten children. It remains to
show how a transition oracle acting truthfully can determine if we
are on the correct path.

We have each node of $A$ and $C$
maintain $v.l$ and $v.r$ pointers as described in \cref{sec:BST-appendix}.
The values at $v.l$ and $v.r$ are ancestors of 
$v$ whose keys are just smaller and just larger 
than $v$ respectively. They determine the
interval of possible values of any node in 
$v$'s subtree, which corresponds to an interval of points 
on each hull. See Figure~\ref{fig:dynamic-hull} for an example.

Say
we are currently at node $a \in A$ and node $c \in C$. We can perform
four case comparisons: $a.r$ with $c.r$, $a.r$ with $c.l$, $a.l$
with $c.r$, and $a.l$ with $c.l$. Each outputs a
region of $A$ and $C$ that the two tangent points must be located
in. 
If the intersection of all four case comparisons
contains the regions bounded by $a.r$ and $a.l$ as well as $c.r$ and $c.l$, then we are on
a valid path by the correctness of Overmars and van Leeuwen's case
analysis~\cite{overmars1981maintenance}. 
Since the noise-free process
takes $O(\log n)$ time, with a path-guided pushdown random
walk using the transition oracle described above, this process
takes $O(\log n)$ time w.h.p.

\begin{figure}
    \centering
    \includegraphics[width=0.55\linewidth]{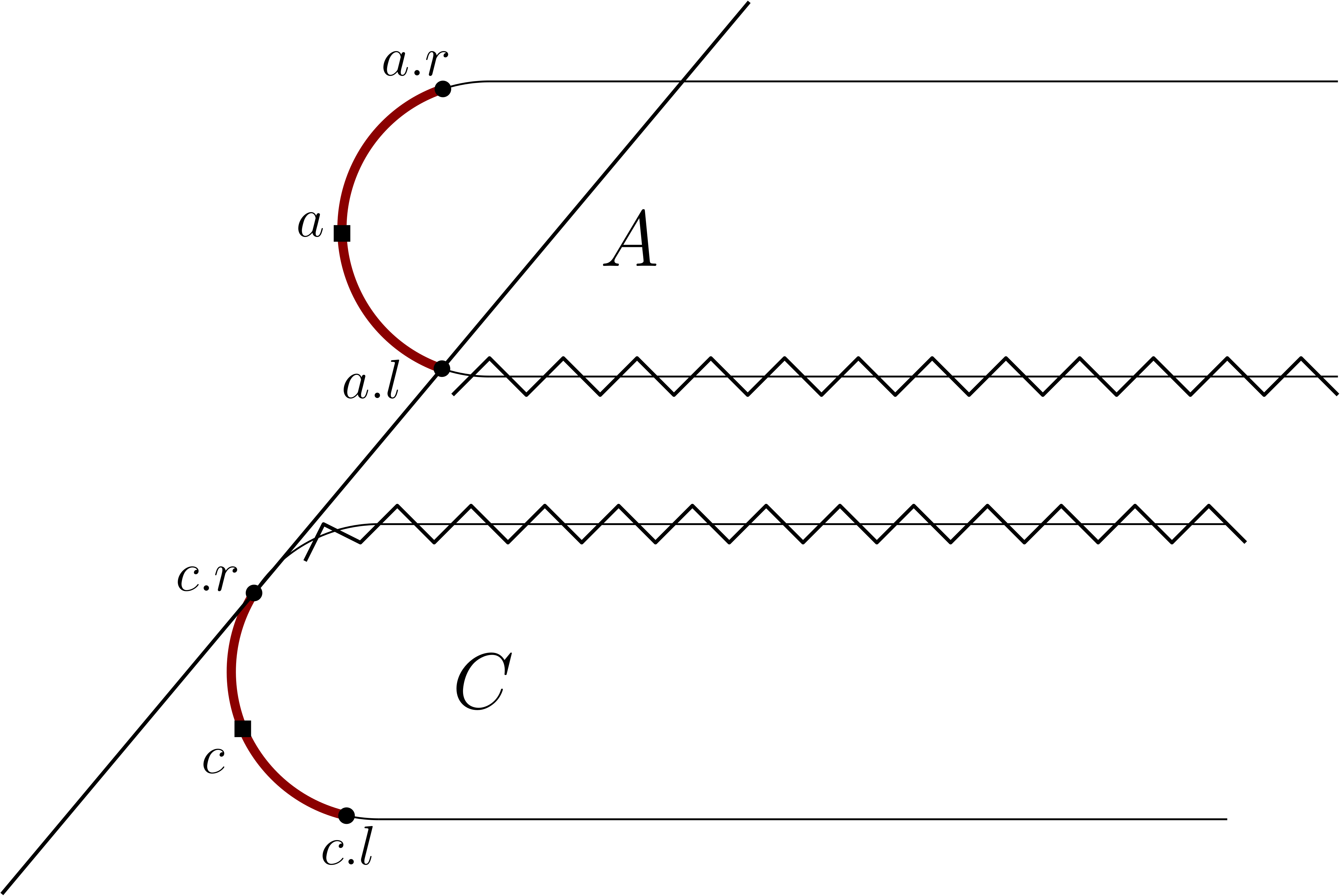}
    \caption{
    Here we show one of four comparisons computed by the transition oracle.
    The points $a$ and~$c$ represent the node we are at in each structure. 
    The highlighted regions represent all possible values in that node's subtree, bordered by points $a.r$ and $a.l$ (resp. $c.r$ and $c.l$).
    This comparison is valid as it does not eliminate the highlighted regions. If all four are valid, 
    then we must be on the right path in our double binary search. 
    }
    \label{fig:dynamic-hull}
\end{figure}

To query a hull, we can perform noisy binary search
on the $lc$- and $rc$-hull structures using the transition oracle
of \cref{sec:BST-appendix}. Updating the
convex hull utilizes the technical lemma discussed above along with
split and join operations on binary search trees
\cite{overmars1981maintenance}. 
We conclude the following.
\begin{theorem}
\label{thm:2D-dynamic}
We can insert and delete points in a planar convex hull in $O(\log^2 n)$ time per update w.h.p., even with noisy primitives.
\end{theorem}
\subsection{Convex Hulls in 3D}\label{sec:3D-hull}

In this section, we show how to construct 3D convex hulls 
in $O(n\log n)$ time w.h.p. even with noisy primitive operations.
The main challenges in this case are first to define an appropriate
algorithm in the noise-free model and then define a good transition
oracle for such an algorithm.
For example, it does not seem possible to efficiently implement
the divide-and-conquer algorithm 
of Preparata and Hong~\cite{preparata1977convex} in the noisy primitive
model as each combine step performs $O(n)$ primitive operations.
See \cref{sec:history-DAG-ex} for an explanation of the 3D hull
random incremental construction we base our solution on. 
We begin
by constructing a tetrahedron through any four points; we need the
orientation of these four points, which can be found by the trivial
repetition strategy in $O(\log n)$ time. The points within this
tetrahedron can be discarded, again using the trivial repetition
strategy in $O(\log n)$ time per point.
\begin{figure}
    \centering
    \includegraphics[width=0.45\linewidth]{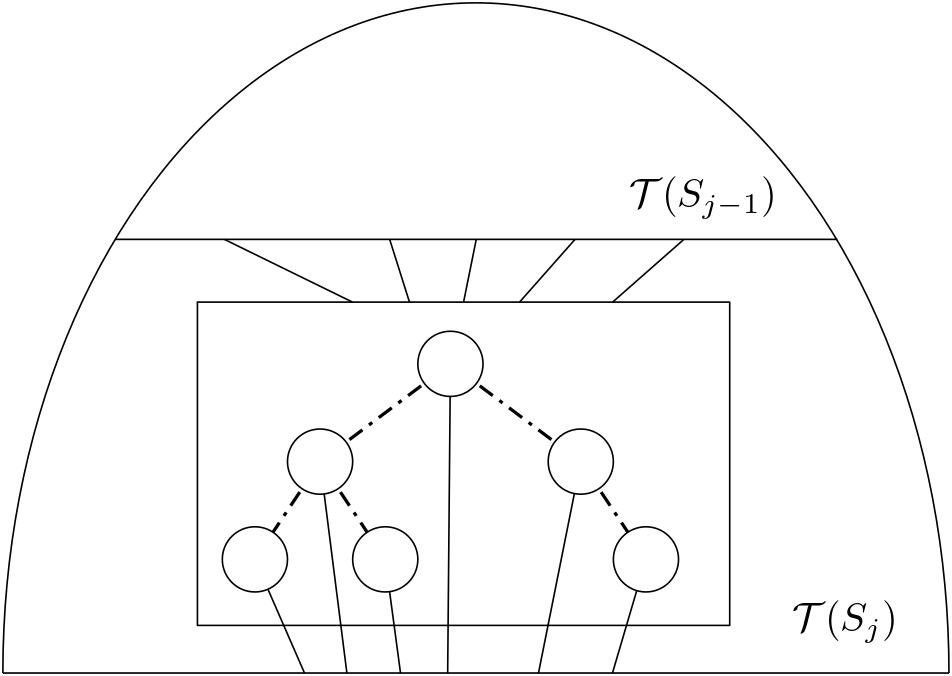}
    \caption{A depiction of how the radial-search structure is embedded into the history DAG. $\mathcal T(S_{j-1})$ represents the history DAG before the $j$th point was added to the hull. Every node that was a leaf in $\mathcal T(S_{j-1})$ now points to each node created in $\mathcal T(S_j)$. 
    Destroyed nodes have pointers to all created nodes. The dashed lines between the nodes in $\mathcal T(S_j)$, represent the connections of the radial search structure, represented as a BST.}
    \label{fig:radial-search}
\end{figure}
The main challenge to designing a transition oracle for the history
DAG in this method is its nested nature. We have the nodes in the DAG themselves
representing facets and pointing to the facets that replaced them
when they were deleted. We also have the radial search structures,
which maintain the set of new facets produced in a given iteration
in sorted clockwise order around the point that generated them. See
\cref{fig:radial-search} for a depiction of this. 

Our transition oracle
needs to be able
to determine if we are on the correct path.
Say we are given a query point, $q_i$, being inserted in the $i$th iteration of the RIC algorithm and are
at node $v$ of the history DAG. To see if we are performing the
right radial search, we first check if the ray from the origin to
$q$ passes through $v$'s parent (in the history DAG, not the radial
search structure). Then we must determine if we are performing
radial search correctly. To do so, we can augment the radial search
tree as described in \cref{sec:BST-appendix} and perform two more
orientation tests to determine if $v$ is on the right path in our
current radial search structure. If a comparison returns false,
then we use our stack to backtrack up the radial search structure
or, if we are at its root, back to the parent node in the history
DAG and its respective radial search tree. Thus, in just three
comparisons, we can determine if we are on the right path in the
history DAG. Similar to our trapezoidal decomposition analysis, we observe that the point-location cost in each iteration is $O(\ell_i + \log n)$, where $\ell_i$ is the length of the unique path in the history DAG corresponding to our query $q$. An analysis by Mulmuley~\cite{mulmuley1994computational} shows that, over all $n$ iterations, $\sum_{i=1}^n \ell_i = O(n\log n)$ in expectation. As a result, total point-location cost is $O(n\log n)$ in expectation.  


Once a conflicting facet is found through the random walk, we can
again walk around the hull and perform the trivial repetition
strategy of \cref{sec:repetition} to determine the set of all conflicting
facets, $X_i$, in $O(|X_i|\log n)$ time.
Each facet must be created before it is destroyed, 
so $\sum_{i=1}^n|X_i|$ is upper-bounded by the total size of the history DAG. 
The same analysis by Mulmuley~\cite{mulmuley1994computational} shows that the expected size of the history DAG is $O(n)$. We conclude the following.
\begin{theorem}
\label{thm:3D-hull}
We can successfully compute a 3D convex hull of $n$ points w.h.p. in expected $O(n\log n)$ time, even with noisy primitives.
\end{theorem}



\section{Delaunay Triangulations and Voronoi Diagrams}\label{sec:DT}

In this section, we describe an algorithm that computes the Delaunay
triangulation (DT) of a set of points in the plane in the noisy-primitives
model. Noise is associated with determining whether a point
$p_i$ lies within some triangle $\Delta$ and whether a point $p_i$ lies
in the circle defined by some Delaunay edge $e$. Here, we
describe the algorithm under the Euclidean metric. 
We later show how to generalize our algorithm 
for other metrics.

Once again, we use a history DAG RIC similar to what we described
in the previous section. This time, each node in the DAG
represents the triangles that exist throughout the construction of
the triangulation (see
\cite{mulmuley1994computational} for details). A leaf node is a
triangle that exists in the current version, and an internal node
is a triangle that was destroyed in a previous iteration. When we
wish to insert a new point $p_i$, we first use the history DAG to locate
where $p_i$ is in the current DT. If $p_i$ is within a triangle $\Delta$,
we split $\Delta$ into three triangles and add them as children of
$\Delta$. Otherwise $p_i$ is on some edge $e$, and we split the two
adjacent triangles $\Delta_1$ and $\Delta_2$ into two triangles
each. Once this is done, we repeatedly flip edges that violate the
empty circle property so that the triangulation becomes a DT again
(see~\cite{decomputational} Chapter 9.3). As described in
\cite{mulmuley1994computational}, we store the new triangles in a radial
search structure sorted CCW around the added point $p_i$.
Using a path-guided pushdown random walk in the history DAG to find
the triangle containing $p_i$ is broadly similar to the process used
to find conflicting faces in the 3D convex hull RIC. Once again,
we use the ray-shooting and radial search queries described in
\cref{sec:3D-hull}. 
With a similar
transition oracle to 3D convex hulls, we can determine where
$p_i$ is in the current DT in $O(\ell_i + \log n)$ time, where $\ell_i$ is the length of $p_i$'s unique path in the history DAG.  
Unlike for 3D convex hulls, 
there is a unique triangle or 
edge that is in conflict with $p_i$, 
which simplifies the process of determining a 
unique path through the history DAG. 
We observe that if $p_i$ is not located in the triangle
$\Delta$ represented by the current node in the DAG, then we are
not on the right path. 
Once we
have found the triangle containing~$p_i$ and added in new edges, 
we begin edge flipping. Using
the trivial repetition strategy, each 
empty circle test costs $O(\log n)$ time. 
The number of empty circle tests is proportional
to the number of edges flipped, which is less than the total number
of triangles created throughout the algorithm. 
Once again, Mulmuley's~\cite{mulmuley1994computational} analysis shows that $\sum_{i=1}^n\ell_i = O(n\log n)$ in expectation and the number of empty circle tests needed to update the DAG after point-location is $O(n)$ in expectation. We conclude the following.
\begin{theorem}
\label{thm:DT}
We can successfully compute a Delaunay triangulation of $n$ points w.h.p. in expected $O(n\log n)$ time, even with noisy primitives.
\end{theorem}

We remark that this also leads to an algorithm for constructing a Euclidean minimum spanning tree, in which we construct the Delaunay triangulation, apply a noisy sorting algorithm to its edge lengths, and then (with no more need for noisy primitives) apply Kruskal's algorithm to the sorted edge sequence. 
In 
the following section,
we generalize our Delaunay triangulation algorithm to work for a wider class of metrics, including all $L_p$ metrics for $1 < p < \infty$. 


\subsection{Generalized Delaunay Triangulations}\label{sec:DT-gen}

In this section, we show the following:

\begin{theorem}
Given $n$ points in the plane, we can successfully compute a generalized Delaunay triangulation generated by homothets of any smooth shape w.h.p. in expected time $O(n\log n)$, even with noisy primitives.
\end{theorem}

A \emph{shape Delaunay tessellation} is a generalization of a Delaunay triangulation first defined by Drysdale \cite{drysdale1990practical}. Rather than a circle, we instantiate some convex compact set $C$ in $\mathbf{R}^2$ and redefine the empty circle property on an edge $\overline{pq}$ like so: given a set of points $S$ and some shape $C$, the edge $\overline{pq}$ exists in the shape Delaunay tessellation $DT_C(S)$ iff there exists some homothet of $C$ with $p$ and $q$ on its boundary that contains no other points of $S$ \cite{aurenhammer2014shape}. 
We will consider smooth shapes, as arbitrary non-smooth shapes may cause $DT_C(S)$ to not be a triangulation under certain point sets \cite{aurenhammer2014shape}. Note that all $L_p$ metrics for $1 < p < \infty$ correspond to smooth shapes (rather than a unit circle, they are unit rounded squares) and furthermore, all smooth shapes produce triangulations of the convex hull of $S$ \cite{aurenhammer2014shape,skyum1991sweepline}. Similar to above, we use the general position assumption that no four points lie on a homothet of $C$.

After inserting a point, $p_r$, in the incremental algorithm, 
we draw new edges that connect~$p_r$ to the points that comprise the triangle(s) that $p$ landed in. 
These new edges are Delaunay as either they are fully enclosed by the homothets of $C$ that covered the triangle(s) $p$ was located in or they are enclosed by the union of previously-empty shapes that covered each edge of the triangles. After this, the algorithm finds adjacent edges, tests if they obey the empty shape property, and flips them if not. 
In Theorem 3 of a work by Aurenhammer and Paulini \cite{aurenhammer2014shape}, the authors prove that for any convex shape $C$, local Delaunay flips lead to $DT_C(S)$ by showing that there exists some lexicographical ordering to these triangulations similar to how the Delaunay triangulation in $L_2$ maximizes the minimum angle over all triangles. Thus, the flipping algorithm used after an incremental insertion will terminate at a triangulation that is $DT_C(P_r)$, where $P_r$ is the set $\{p_1, ..., p_r\}$. 

Because the flipping process works as before, the point location structure can be easily adapted to this setting. As stated above, using a smooth shape $C$ guarantees that $DT_C(S)$ is an actual triangulation as opposed to a tree. Thus, in every iteration, there exists a triangle that the new point is enclosed by. The history DAG data structure is agnostic to the specific edge flips being made during the course of the algorithm and so works as expected. Our transition oracle also behaves the same.

In the Euclidean case, the backwards analysis of~\cite{decomputational} uses the fact that a Delaunay triangulation of $r$ points has $O(r)$ edges, any new triangle created in iteration $r$ is incident to the newly inserted point $p_r$, and triangles are defined by at most three points. By planarity and by construction of the algorithm, all three hold true in this generalized case. Thus, we have the desired result. 

\section{Discussion}
We have shown that a large variety of computational
geometry algorithms can be implemented in optimal time or, for dynamic 2D hull, in time that matches
the original work's complexity
with high probability
in the noisy comparison model. We believe that this is the first
work that adapts the techniques of noisy sorting and searching to
the classic algorithms of computational geometry,
and we hope that it inspires work in other settings,
such as graph algorithms.






\bibliography{refs}

\appendix
\section{Path-guided Pushdown Random Walks}\label{sec:dag-lemma-pf}

\subsection{Correctness and analysis}
In this section we prove that w.h.p. the path-guided random walk terminates with the correct answer to its search problem (\cref{thm:path}).
To do so we first need to define a termination condition for the walk. It is not adequate to terminate merely after an appropriate number of steps: with constant probability, the final step of the walk will be taken after an erroneous oracle result, and may be incorrect. On the stack used to guide the algorithm, we store along with each vertex a \emph{repetition count}, equal to one if the vertex is different from the previous vertex on the stack, and equal to the previous repetition count otherwise. We terminate the algorithm when this repetition count reaches an appropriately large value, $\Theta(\log(1/\varepsilon))$.

\begin{lemma}
\label{lem:whp-goal}
If it does not terminate earlier, the path-guided random walk will reach the correct goal vertex $t$, within $\Theta(|P|+\log(1/\varepsilon))$ steps, with high probability, with a constant factor determined by the analysis below.
\end{lemma}

\begin{proof}
Define a call to the transition oracle, $T$, in our random walk as ``good'' if 
it returns the correct answer (i.e., $T$ does not lie)
and ``bad'' otherwise, so that each call is bad independently 
with probability at most $p_e$.
Note that if we are at a node, $v\in P$, and $v\not= t$, then 
a good call with either undo a previous bad call to stay at $v$
or it will move to the next vertex in $P$.
Also, if we are already at $t$, then a good call will keep us at $t$,
adding another copy of $t$ to the top of the stack, $S$.
Alternatively, if we are a node, $v\not\in P$, then
a good call with either undo a previous bad call to stay at $v$
or it will move back to a previous node in our traversal, which undoes
a previous bad call to move to $v$. Moreover, if we repeat this latter case,
we will eventually return back to a node in $P$.
Thus, every good call either undoes a previous bad call or makes progress
towards the target vertex, $t$.
Admittedly, if we 
are at a node, $v\in P$, then 
a bad call can undo a previous good call (e.g., which was to move to $v$ or stay
at $v$ if $v=t$). Also, a sequence of bad calls can even deviate from $P$ and 
return back to it---but because $G$ is a DAG, a series of bad calls cannot
return back to a previously visited vertex of $P$.
Thus, a path-guided pushdown random walk will successfully 
reach the target vertex, $t$, if the difference
between the number of good calls and bad calls is at least $|P|$,
the length of the path, $P$.
Let $X_i$ be an indicator random variable that is $1$ iff the $i$th call
to the transition oracle
is bad, and let $X=\sum_{i=1}^{N} X_i$.
Since each call is bad independent of all other calls,
we can apply a Chernoff bound
to determine an upper bound on the probability that
the difference
between the number of good calls and bad calls is less
than $|P|$, i.e., if $(N-X)-X<|P|$, that is, if $X>(N-|P|)/2$.
Further, note that $\mu=E[X]=p_e N\le (1/15)N$.
Then, for $N=3(|P|+\log (1/\varepsilon))$,
the failure probability is
\begin{eqnarray*}
\Pr\left(X > \frac{N-|P|}{2}\right) 
&=& 
\Pr\left(X > |P|+\frac{3}{2}\log ({1}/{\varepsilon})\right) \\ 
&=&
\Pr\left(X > \frac{1}{3}\left(3|P|+\frac{9}{2}\log ({1}/{\varepsilon})\right)\right) \\ 
&\le&
\Pr\left(X > \frac{1}{3}N\right) .
\end{eqnarray*}
Further, by a Chernoff bound from Dillencourt, Goodrich,
and Mitzenmacher~\cite{dillencourt} (Theorem~7),
\[
\Pr\left(X > \frac{1}{3}N\right) < 2^{-N/3} \le 2^{-\log (1/\varepsilon)} = {\varepsilon}.
\]
This establishes the proof.
\end{proof}

As mentioned above, the condition $p_e< 1/15$ can, with
the same asymptotic performance as in \cref{thm:path}
be replaced with any error probability bounded away from $1/2$.
We caution, however,
that while \cref{thm:path} provides a high-probability guarantee we
reach the target vertex, $t$, it does not provide a high-probability
guarantee we visit every vertex of the path, $P$.
There is an exception to this, 
which we describe now.

\begin{corollary}\label{cor:path-tree}
If $G$ is a tree, then the 
path-guided pushdown random walk in $G$
will visit every node in $P$.
\end{corollary}
\begin{proof}
By definition, there is a unique path $P^*$ 
from the root of $G$ to the target vertex $t$. 
Thus, according to \cref{thm:path}, 
with high probability after $N$ steps, the path represented by the stack $S$ must equal 
the intended path $P^*$, which also must equal $P$.
\end{proof}

\begin{lemma}
\label{lem:terminate}
If it does not terminate earlier, the path-guided pushdown random walk will accumulate $\Theta(\log(1/\varepsilon))$ copies of the goal vertex $t$ on its stack, after $\Theta(|P|+\log(1/\varepsilon))$ steps, and therefore terminate, with high probability.
\end{lemma}

\begin{proof}
This follows by applying \cref{lem:whp-goal} to a modified DAG $G'$ in which we expand each vertex $v$ of $G$ into a chain of copies $(v,1)$, $(v,2)$, $(v,3)$, $\dots$ of $v$ with a repetition count, as used in the termination condition of the algorithm.
\end{proof}

\begin{lemma}
\label{lem:incorrect}
The path-guided pushdown random walk will not terminate with any other vertex than $t$, with high probability.
\end{lemma}

\begin{proof}
By another application of Chernoff bounds, and by the assumption that $\varepsilon$ is polynomially small, the probability that the algorithm terminates at step $i$ with an incorrect vertex is itself polynomially small, with an exponent that can be made arbitrarily large (independent of our original choice of $\varepsilon$) by an appropriate choice of the constant factor in the $\Theta(\log(1/\varepsilon))$ repetition count threshold used for the termination condition. By choosing this constant factor appropriately, we can make the probability of termination at any fixed step smaller than $\varepsilon$ by a factor at least as large as the high-probability bound on the number of steps of the walk in \cref{lem:terminate}. The result follows by applying the union bound to all the steps of the walk.
\end{proof}

\begin{proof}[Proof of \cref{thm:path}.]
Termination of the algorithm w.h.p. follows from \cref{lem:terminate}, and correctness once terminated follows from \cref{lem:incorrect}.
For a given high-probability bound $1-\varepsilon$, we apply \cref{lem:terminate} and \cref{lem:incorrect} with $\varepsilon/2$ and then apply the union bound to get a bound on the probability that the algorithm terminates with a correct answer.
\end{proof}

\subsection{Counterexample to Generalizations}

It is very tempting to consider a generalization of the path-guided pushdown random walk defined by a set of \emph{valid nodes} and \emph{goal nodes}, with the property that a non-noisy search starting from any valid node will reach a goal node and then stop. With a noisy oracle that determines whether a node is valid and if so follows a noisy version of the same search, one could hope that this generalized algorithm would quickly reach a goal state. The path-guided pushdown random walk would then be a special case where the valid nodes are exactly the nodes on the non-noisy search path from the root node. However, as we show in this section, this generalization does not work with the same fast high-probability termination bounds, unless additional assumptions are made (such as the assumptions giving the special case of the path-guided pushdown random walk).

Consider the following search problem: the DAG to be searched is simply a complete binary tree with height $\log_2 n$. One root-to-leaf path on this DAG is marked as valid, with the nodes on this path alternating between non-goal and goal nodes, so that the non-noisy search from the root stops after one step but other later stops would also produce a valid result. We have a noisy oracle with the following behavior:
\begin{itemize}
\item At a leaf node of the binary tree, or an invalid node, it always produces the correct result.
\item At a valid non-leaf node, it follows the same behavior as a non-noisy oracle with probability $14/15$: that is, it correctly identifies whether the node is a goal, and if not returns the unique valid child. With probability $K\log\log n/\log n$ (for some suitably large constant $K$) it returns the valid child regardless of whether the node is a goal. And with the remaining probability $1/15 - K\log\log n/\log n$ it returns the invalid child.
\end{itemize}

Now consider the behavior on this problem of a path-guided random walk, as defined above with a repetition-count termination condition. It will follow the valid path in the binary tree, with brief diversions whenever the noisy oracle causes it to walk to a node not on the path. At any goal node, it will wait at that node, accumulating more repetitions of the node, until either it achieves $C\log n$ repetitions (with a constant factor $C$ determined by the desired high probability bound) or the noisy oracle tells it to take one more step along the valid path. The probability of reaching $C\log n$ repetitions without taking one more step is $(1-K\log\log n/\log n)^{C\log n}\approx\exp -(K/C)\log\log n$, and by adjusting $K$ relative to $C$ we can make this probability less than $1/(2\log n)$. With this choice, by the union bound, the random walk will eventually take one more step from each goal node until finally reaching the leaf goal node. The expected number of steps that it waits at each goal node will be $\Theta(\log n/\log\log n)$, so the expected number of steps in the overall random walk, before reaching the leaf goal node and then accumulating enough repetitions to terminate, will be $\Theta(\log^2n/\log\log n)$.

It also would not work to use a termination condition of making enough steps to have high probability of reaching a goal node, and then stopping. If this number of steps is not enough to reach the leaf goal node, the probability of stopping at a non-goal node would be too high.

Thus, without either a change of algorithm or more assumptions on the valid and goal node sets, this generalized random walk can be made to take a number of steps that is  longer than $|P|+\log 1/\varepsilon$ by a non-constant factor.
\section{Transition Oracle for Binary Search Trees}
\label{sec:BST-appendix}

While path-guided pushdown random walks applies to many DAGs, 
many fundamental computational geometry algorithms rely on binary search trees. 
In this appendix, we present a transition oracle for BSTs. 


Say that we are attempting to find some value $x$ in a binary search tree $T$. By 
properties of BSTs, 
there is a unique path from the root to the node containing $x$. 
If there is no such node, there is still a unique path to the nonexistent leaf that would contain $x$.
When a call to the transition oracle $T$ is good, 
we expect it to correctly determine if we are on said valid path only using a constant number of comparisons. 
To do so, we add a constant amount of extra data to each node. 

For each node $v$ of the BST, 
let $P_v$ be the unique path from the root to $v$. 
Let $v.l$ be the lowest ancestor of $v$ whose right child is in $P_v$ (or a sentinel representing $-\infty$). 
Likewise, let $v.r$ be the lowest ancestor of $v$ whose left child is in $P_v$ (or a sentinel representing $\infty$). 
One of $\{v.l, v.r\}$ is $v$'s parent. 
Observe that, if we have correctly navigated to $v$,
$x$ must be in between the values held by $v.l$ and $v.r$. 
Thus, by adding two pointers to each node, 
the transition oracle can determine if we are on the correct path using two comparisons. 
See \cref{fig:bst-oracle} for a visual depiction of the pointers $v.l$ and $v.r$ and their effect.
Augmenting the tree with these pointers can be done in $O(\log n)$ extra time per insertion and deletion. 
With slight modification, we can also support inexact queries such as for predecessors or successors. 

\begin{figure}
    \centering
    \includegraphics[width=0.65\linewidth]{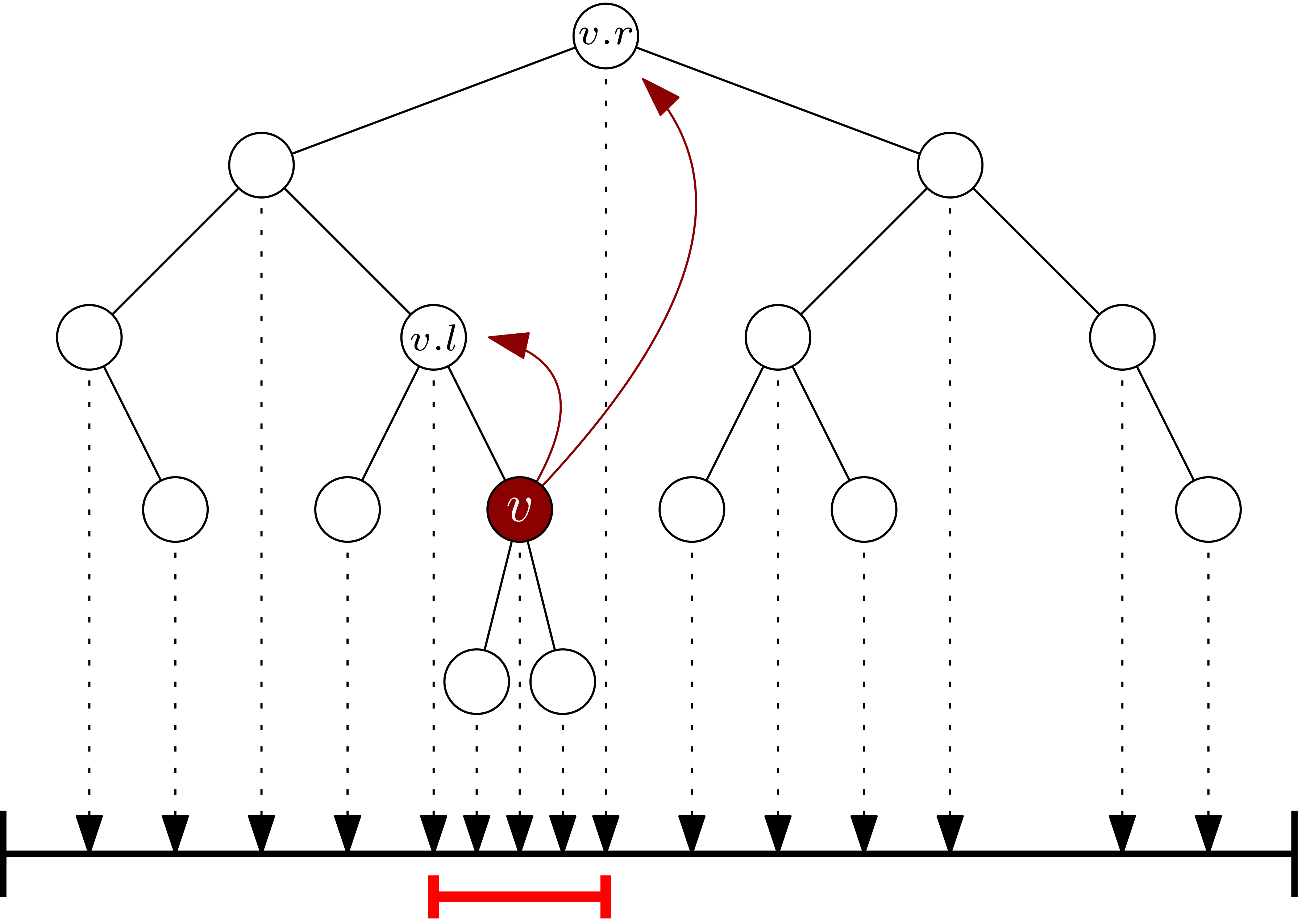}
    \caption{Here we have depicted the $v.l$ and $v.r$ pointers for the node $v$. Notice that $v.l$ is $v$'s first ancestor whose value is smaller than $v$'s. Likewise, $v.r$ is its first ancestor whose value is larger than $v$'s. The dotted lines map nodes to their relative ordering on a number line. It is clear that the values held by $v.l$ and $v.r$ form the interval of possible values that could be held by nodes in $v$'s subtree. If the query is located strictly in this interval and if it is in the tree at all, then it must be within $v$'s subtree. If the query is not in this interval, then it cannot be in $v$'s subtree and we have made an errant comparison earlier. }
    \label{fig:bst-oracle}
\end{figure}

\section{Non-Noisy Random Incremental Constructions}\label{sec:RIC-expl}
In this section we summarize known work on random incremental constructions, used as the basis for some of our results.

\subsection{Trapezoidal Decomposition with History DAG}
In the trapezoidal decomposition problem, we wish to decompose the space based on a set of non-intersecting line segments in the plane. By drawing vertical lines at each segment endpoint, we can split the space into trapezoids (some of which are unbounded). To do this efficiently, we permute the set of segments and insert them one at a time. Our method of solving this problem, which we will see again in \cref{sec:3D-hull} and \cref{sec:DT} is to consider a random incremental construction that uses a \emph{history DAG}. Each node of the structure represents a trapezoid. Leaves represent trapezoids that exist in the current decomposition. Internal nodes represent trapezoids destroyed by a segment.  

For simplicity, we imagine that a bounding box exists that contain all line segments. We also only consider the case of non-intersecting segments. The initial history DAG is a single node that represents the bounding box. Inserting the first segment splits the space into four trapezoids, which each become a child of the initial node in the history DAG. Likewise, each iteration of the RIC, each destroyed trapezoid points to the one or more trapezoids that replace it. Additional pointers connect trapezoids that are adjacent in the trapezoidal map, so that we can walk around the structure easily. 

It remains to describe how to use the history DAG to determine what trapezoids a new segment destroys, and what new ones to create. To insert a segment, we take both of its endpoints and navigate down the history DAG using them. Each node in the history DAG contains at most four children that cover the parent trapezoid but do not overlap. Thus, when navigating down the DAG, there is always a unique child that contains the point we are searching for. When we have found the leaf trapezoids that contain both endpoints, we walk along the segment that connects the endpoints, destroying and merging trapezoids as needed. The work spent on the algorithm then revolves around navigating down the tree and adding nodes. 

Using a backwards analysis, we can bound the costs associated with modification of the tree and search. 
It is known that a trapezoidal decomposition of $n$ segments has at most $O(n)$ trapezoids~\cite{decomputational}. 
Also, observe that each trapezoid is defined by at most four segments. 
Thus, each trapezoid in our final decomposition would not exist in the previous iteration 
if any of its supporting segments were inserted last. 
The probability one of those segments was added last is $\leq 4/n$,
meaning this trapezoid was created this iteration with probability $\leq 4/n$. 
There are $O(n)$ trapezoids, so the expected number of 
trapezoids created this iteration is at most $O(n) \times 4/n = O(1)$. 
Over all $n$ iterations, our DAG is of size $O(n)$ in expectation. It follows that the number of destroyed trapezoids is at most $O(n)$ in expectation as we can only destroy trapezoids after they are created. 

We use a similar analysis to bound point-location cost. 
Consider some fixed point $p$ whose location is independent 
of previously inserted segments. 
Say we wish to find $p$'s location in the final trapezoidal decomposition. 
In iteration $i$ of the algorithm, the probability that the addition of the new segment 
destroys the trapezoid $\Delta$ that previously housed $p$ is $\leq 4/i$ 
for similar reasons to what was described above. 
If this occurs, 
then $p$ must be located in exactly one of the child
trapezoids produced from $\Delta$. 
This increases the history DAG point location cost for $p$ by 1.
Thus, total point location cost for $p$ in the final decomposition is at most 
$\sum_{i=1}^n O(1)\times 4/i = O(\log n)$ in expectation. 
We perform point location $2n$ times over the course of the algorithm. Performing point-location at iteration $i$ will take at most $O(\log i) \subseteq O(\log n)$ time in expectation. Therefore, total point location cost is at most $O(n\log n)$ in expectation.

\subsection{3D Convex Hull with  History DAG}\label{sec:history-DAG-ex}
We consider the RIC with a history DAG described by Mulmuley \cite{mulmuley1994computational} for Voronoi diagrams, later adapted to 3D convex hulls. 
In this search structure, each node represents a facet. Root nodes represent each facet of the initial tetrahedron. Each leaf of the structure represents a facet of the current hull. Upon inserting a new point, we delete a set of facets $X$ and insert a new set of facets $Y$. In the DAG, this is represented by initializing a new node for every $f \in Y$. Additionally, every node representing each $f \in X$ has a pointer to each node representing every $f \in Y$.

To obtain faster search times, we also organize the new nodes representing $Y$ into a \emph{radial search structure}. Let $p$ be the newly added point. As will be clear soon, we store the new edges incident to $p$ in the radial search structure rather than the faces. These edges are ordered cyclically around $p$. Convert this cyclic ordering into a linear ordering by arbitrarily picking the ``least clockwise'' edge and organize them into a search structure (e.g. sorted array, binary search tree, skip list, etc.). This structure is then embedded into the history DAG along with the nodes for $Y$.


Now we show how to query this structure. After we create the initial tetrahedron, we choose a coordinate system such that its origin lies in the tetrahedron (thus the origin will be located within the current hull at every step of the RIC). Imagine shooting a ray from the origin to our query point $q$. We start our query at the node representing the facet of the tetrahedron which the ray pierces. Likewise, we descend to the child that represents the next facet whose interior our ray passes through, and so on until we reach a leaf node. By general position, the ray only passes through the interior of a facet, never the edge or at a vertex. In this way, we have a query that designates a unique path down the history DAG, which is required to apply path-guided pushdown random walks.

To make this process more efficient, we rely on the radial search structures. Say we are at some node $v$ of the history DAG and are querying a radial search structure to find the unique child to descend to. Let $p$ be the point that created the edges in the current radial search structure and let line $\overrightarrow{Op}$ be a ray that starts at the origin and passes through $p$. From the edges incident to $p$, we define a system of half-spaces, all of which are bounded by $\overrightarrow{Op}$ (think of them as like wedges of an orange surrounding $\overrightarrow{Op}$ where each plane that slices the wedges is drawn between an edge incident to $p$ and $\overrightarrow{Op}$). The goal of this search is to determine which face our query ray $\overrightarrow{Oq}$ pierces. Each node of the radial search structure represents one of the half-spaces in the system. An orientation test determines which side of the half-plane the query $q$ is, and so determines which direction to recurse in the tree. At the end of the search, we will have found two adjacent half-spaces that bound $q$ on either side. Because the edges incident to $p$ defined these half-spaces, this corresponds exactly to the face that the ray $\overrightarrow{Oq}$ pierces. Thus, we have found the unique child of $v$ that we were looking for. 


We use the history DAG to determine the set of facets $X_i$ that conflict with our new point $q$ each iteration. Notice that the above query only finds one such facet. However, given a single conflicting facet, we can simply walk around the hull and perform orientation tests on adjacent facets to recover the set of conflicting facets in $O(|X_i|)$ time in the non-noisy setting. Let $\ell_i$ be the length of the path that $q$ took in the history DAG to find a conflicting facet. Then it takes $O(\ell_i + |X_i|)$ time to update the history DAG. An analysis by~\cite{mulmuley1994computational} (also see~\cite{decomputational}) shows that $\sum_{i=1}^n \ell_i = O(n\log n)$ and $\sum_{i=1}^n |X_i| = O(n)$ in expectation. Then the total runtime of the algorithm in the non-noisy setting is $O(n\log n)$ in expectation.


\section{Lower Bounds}
\label{sec:lower-bounds}
In this section we prove lower bounds on computing with noisy primitives showing that for finding the closest pair of points (which takes $O(n)$ time in the non-noisy setting in a model of computation allowing integer rounding) and for constructing line arrangements (which takes time $O(n^2)$ in the non-noisy setting) it is not possible to avoid the logarithmic time penalty imposed by the trivial repetition strategy.

\subsection{Closest Points}
\label{sec:closest-lb}

In this section we prove that computing the closest point w.h.p. with noisy primitives requires $\Omega(n\log n)$ time; thus, our $O(n\log n)$ time algorithm is optimal, despite the existence of non-noisy closest point algorithms taking time $O(n)$. For this section, we make no assumption that the computation can be modeled as a comparison tree or decision tree (unlike past lower bounds for noisy minimum-finding). Instead, we assume merely that the only direct access to the input data is through noisy primitives. We assume that each primitive takes as input $O(1)$ data points and produces an arbitrary value (not necessarily Boolean) as output, which is correct with probability $1-p$ and incorrect with probability $p$, for some constant error probability $p<\tfrac12$. For our lower bound, we restrict the incorrect values to be values that could have been produced by a non-noisy version of the primitive on different data points; restricting the model in this way only makes the lower bound stronger. Once a value has been returned from a noisy primitive, the algorithm is free to perform arbitrary computation on it. However, with this restriction, incorrect values may be chosen adversarially.

\begin{theorem}
\label{thm:closest-lb}
Let $c>0$. Then in the model of computation described above with constant error probability $p$, computing a closest point among $n$ 2D points, with probability $\ge 1-n^{-c}$, requires calls to $\Omega(n\log n)$ noisy primitives, even for the expected number of calls made by a randomized algorithm, and
even for noisy primitives that (when erroneous) produce an answer that would be valid for an arbitrarily small perturbation of the input relative to the closest pair distance.
\end{theorem}

\begin{proof}
Let $n$ be any multiple of four, and consider any point set $S$ in general position, grouped into $n/2$ pairs of points, each pair approximately at unit distance (as our general position assumption allows), with all other distances larger. We define a hard distribution on random instances $R$ to the closest pair problem by choosing one of the $n/2$ pairs of points uniformly at random and perturbing this pair to be closer than all other pairs. Additionally, we specify a noisy primitive that, with probability $p$, answers with the result that would be correct for $S$ instead of for the perturbed input $R$. By choosing the perturbations appropriately, this erronous result will be a result that would be valid for an arbitrarily small perturbation of the input.  We will first show that, for this input distribution and this adversary, a deterministic algorithm requires $\Omega(n\log n)$ noisy primitives to solve the problem correctly with the given probability. 

Let $t$ denote the maximum number of data points that participate in a single primitive.
For some suitably small constant $\kappa$, we have $p^{\kappa\log n}>2n^{-c}$: that is, if we make $\kappa\log n$ calls to noisy primitives, there is a somewhat large probability that all of them produce erroneous outputs. Now suppose that a deterministic algorithm makes at most $\tfrac{\kappa}{4t}n\log n$ calls to noisy primitives for any $n$-point input, and consider the sequence of calls that it would make for point set $S$; this sequence is deterministic as our noise model will not affect the result of any primitive on input $S$. Among the $n/2$ unit-distance pairs of points in $S$, let $F$ be the set of pairs whose two points are involved in at most $\kappa\log n$ calls to primitives in this sequence of calls, and such that this pair is not the output of the algorithm on input $S$. The average number of calls that involve at least one point from a pair is at most $\tfrac{\kappa}{2}\log n$ per pair, and by Markov's theorem more than half of the pairs have a number of calls that is at most twice the average. After removing the output of the algorithm on $S$ from this set of pairs with few calls, we have that $|F|\ge n/4$.

For the two events such that (1) our hard distribution chooses a pair in $F$ to make closest and (2) the first $\kappa\log n$ calls involving the two points from this pair are all noisy, the first event happens with probability $\ge 1/2$. After conditioning on the choice of pair, the second event happens with probability $>2n^{-c}$. Therefore, both events happen with probability $>n^{-c}$.  For all such combinations of events, the deterministic algorithm will receive the same results of primitives as it would with input $S$, and will produce the same output as it would with input $S$, but this output is not in $F$. Thus, the algorithm will be mistaken with probability $>n^{-c}$. Because the algorithm was arbitrary, every deterministic algorithm fails to have high probability of correctness on this input distribution.

Yao's principle~\cite{Yao77} allows us to convert this lower bound on random inputs and deterministic algorithms into a lower bound on random algorithms, as stated in the theorem. This principle is closely related to the minimax principle for zero-sum games and the duality principle for linear programs. It states that, for arbitrary definitions of the cost of an algorithm, the minimum expected cost for a worst-case random distribution on inputs and a resource-bounded deterministic algorithm chosen for that distribution, equals the minimum expected cost for a random algorithm with the same resource bound against a deterministic input chosen for that algorithm. It follows that any valid limit on the performance that can be achieved, proven by finding an input distribution that prevents deterministic algorithms from achieving any better performance, applies as well to the performance that can be achieved by a random algorithm on its worst-case input.

The probability of an incorrect result of an algorithm with noisy primitives is the expected cost, for a cost that is 0 when the algorithm is correct and 1 when it is incorrect. By Yao's principle, applied to the input distribution described above, any randomized algorithm that always makes at most $\tfrac{\kappa}{4t}n\log n$ calls to noisy primitives has a worst-case input causing its probability of an incorrect result to be too high. We can extend this limitation to algorithms whose number of calls to primitives is itself random, as follows.  If a randomized algorithm has an expected number of calls that is at most $\tfrac{\kappa}{8t}n\log n$, then by Markov's inequality it has probability $\ge 1/2$ of making a number of calls that is at most $\tfrac{\kappa}{4t}n\log n$, and therefore of making too few calls to achieve high-probability correctness. By adjusting the parameters of the argument we can ensure that it has probability $\ge 1/2$ of making so few calls that it is incorrect with probability $\ge 2n^{-c}$, so that even if it is perfectly correct in the cases where it makes more than twice its expected number of calls, it will still fail to achieve the given high-probability bound.
\end{proof}

The same lower bound applies to the problem of detecting whether a system of unit disks has a pair of disks that intersect, by scaling the same hard distribution so that, for unit disks centered at its points, the only pair that can intersect is the randomly selected closest pair.

\subsection{Line Arrangements}
\label{sec:arrangement-lb}

In this section we prove a lower bound on the construction of line arrangements. The standard primitive needed for this construction, in the non-noisy case, is an \emph{orientation test}, which for three lines determines whether they meet at a single point or, if not, how they are ordered around the triangle that they form. (This is the projective dual to an operation that takes as input three points and determines whether they are collinear or, if not, how they are ordered around the triangle that they form, and constructing the arrangement is equivalent by projective duality to the problem of determining the order type of a set of $n$ points~\cite{Epp-18}.) Unlike the lower bound for closest pairs, which was agnostic about the primitives used, our lower bound for arrangements will be specific to this primitive.

\begin{theorem}
\label{thm:arrangement-lb}
Constructing an arrangement of $n$ lines with high probability, using a noisy three-line orientation test, requires $\Omega(n^2\log n)$ calls to this test, even for the expected number of calls made by a randomized algorithm.
\end{theorem}

\begin{proof}
Erickson and Seidel~\cite{EriSei-DCG-95} construct an arrangement $A_n$ of $n$ lines, no three having a common intersection, for which there are $\tfrac19n^3+O(n)$ \emph{collapsible triangles}, triples of lines whose orientation can be changed so that they have a common intersection or so that their triangle is reversed without changing the results of any other orientation test. We define a hard input distribution for the problem of constructing arrangements by choosing one of the collapsible triangles in $A_n$ and reversing it, and using a noisy orientation test that, when erroneous, returns the orientation in $A_n$ instead of in this perturbed orientation.

Then, as in the proof of \cref{thm:closest-lb}, consider any deterministic algorithm running on this input distribution, and let $\Delta_0$ be the collapsible triangle (if one exists) that this algorithm would determine to have been reversed when all its queries are answered correctly for $A_n$ instead of for the perturbed input.
For any constant $c$, if the algorithm performs at most $\tfrac1{18}cn^2\log n$ calls to the noisy orientation test, then at least half of the collapsible triangles in $A_n$ other than $\Delta_0$ will be tested at most $c\log n$ times. This leads to probability $\Omega(1/n^c)$, that all of the tests to the reversed triangle are erroneous and that the algorithm will fail to correctly identify the reversed triangle. By adjusting $c$ this failure probability can be made greater than any high-probability bound; thus, no deterministic algorithm can succeed with high probability against this input distribution. An application of Yao's principle extends this impossibility result from deterministic to randomized algorithms. We omit the detailed calculations as they are the same in essential respects as for \cref{thm:closest-lb}.
\end{proof}

\begin{theorem}
Determining whether a given set of $n$ points has three collinear points w.h.p., using a noisy three-point orientation test, requires $\Omega(n^2\log n)$ calls to this test, even for the expected number of calls made by a randomized algorithm.
\end{theorem}

\begin{proof}
The proof is essentially the same as \cref{thm:arrangement-lb} except that we collapse one collapsible triangle rather than reversing it, and then construct the projective dual system of points. In its most basic form, the dual of the construction of Erickson and Seidel~\cite{EriSei-DCG-95} has many collinearities separate from the collapsible triangles (it consists of $n/3$ points on each of three parallel lines, with the collapsible triangles formed by triples of one point from each of these three lines). However, as Erickson and Seidel describe, it can be modified by replacing these lines by shallow convex curves to eliminate these extra collinearities while still preserving the collapsibility of the collapsible triangles.
\end{proof}

For constructing line segment arrangements, additional primitives are required to determine whether two line segments cross or whether one ends before reaching the point where it would cross another segment. For instance, this may be done with orientation tests for triples of endpoints of the segments, as well as orientation tests of triples of lines passing through segments. The same lower bound applies, to a line segment arrangement formed by intersecting Erickson and Seidel's arrangement $A_n$ with a large bounding disk, as long as these additional primitives are \emph{uninformative}: they do not supply any information about whether any collapsible triangle has been reversed. For instance, orientation tests on line segment endpoints are uninformative: both $A_n$ and its perturbation, after intersection with a large bounding disk, have segment endpoints in the same convex position, so the result of an endpoint orientation test does not change when any collapsible triangle is reversed. Therefore, constructing line segment arrangements with both line and point orientation tests may require $\Omega(n^2\log n)$ time in the noisy setting.

\end{document}